\documentclass[reprint,nofootinbib,aps,pra]{revtex4-2}

\usepackage{graphicx}
\usepackage{nicefrac}
\usepackage{dcolumn}
\usepackage{bm}
\usepackage{cancel}
\usepackage{hyperref}
\usepackage{amssymb}
\usepackage{subcaption}
\usepackage{comment}
\usepackage{mathtools}
\usepackage{empheq}
\usepackage{amsthm}
\usepackage{bbm}
\usepackage{mathrsfs}
\usepackage{amsfonts}
\usepackage{color} 
\usepackage{pbox}
\usepackage{xcolor} 
\usepackage{braket}
\usepackage{caption}
\usepackage{enumerate}
\usepackage{physics}

\hypersetup{
    colorlinks=true,
    linkcolor=blue,
    filecolor=magenta,      
    urlcolor=cyan,
}
\usepackage{scalerel,stackengine}
\usepackage{soul}
\soulregister\cite7
\soulregister\eqref7

\usepackage[OT2,T1]{fontenc}
\DeclareSymbolFont{cyrletters}{OT2}{wncyr}{m}{n}
\DeclareMathSymbol{\Sha}{\mathalpha}{cyrletters}{"58}

\usepackage{cleveref}

\makeatletter 

\renewcommand\onecolumngrid{
\do@columngrid{one}{\@ne}%
\def\set@footnotewidth{\onecolumngrid}
\def\footnoterule{\kern-6pt\hrule width 1.5in\kern6pt}%
}

\renewcommand\twocolumngrid{
        \def\footnoterule{
        \dimen@\skip\footins\divide\dimen@\thr@@
        \kern-\dimen@\hrule width.5in\kern\dimen@}
        \do@columngrid{mlt}{\tw@}
}%

\makeatother

\newcommand{\be}{\begin{equation}}
\newcommand{\ee}{\end{equation}}
\newcommand{\ben}{\begin{equation*}}
\newcommand{\een}{\end{equation*}}

\newcommand{\Zb}{\mathbb{Z}}
\newcommand{\Hc}{\mathcal{H}}

\newcommand{\Dc}{\mathfrak{D}}
\newtheorem{theorem}{Theorem}    
\newtheorem{proposition}{Proposition} 
\newtheorem{lemma}{Lemma}
\newtheorem{definition}{Definition}

\begin{document}

\title{Recurrence Time for Finite Quantum Systems}

\author{Chaitanya Gupta}
 \email{chai.gupta@bristol.ac.uk}
\author{Anthony J. Short}
 \email{Tony.Short@bristol.ac.uk}
\affiliation{
 H.H. Wills Physics Laboratory, University of Bristol, Tyndall Avenue, Bristol BS8 1TL, U.K}
\date{\today}

\newcommand{\add}[1]{\textcolor{teal}{\ul{#1}}}
\newcommand{\remove}[1]{\textcolor{red}{\st{#1}}}
\newcommand{\addmath}[1]{\textcolor{teal}{#1}}
\newcommand{\removemath}[1]{\textcolor{red}{\cancel{#1}}}

\graphicspath{ {./images/} }

\begin{abstract}
We study the time it takes for all states of a finite quantum system to  return simultaneously to their original configuration. In particular, we define the recurrence time for a quantum system to be the time at which all time-evolved states are close to their initial configuration, and at least one state has deviated significantly during this interval. Considering finite-dimensional quantum systems evolving unitarily, we find bounds on this notion of recurrence time, for continuous time and discrete time, by using Dirichlet's approximation theorem. We show how the problem of finding a bound on recurrence time can be related to approximating the difference of real numbers by rationals. We present a mathematical result on the latter, which we then use to obtain tighter bounds on recurrence time.
\end{abstract}

\maketitle

\section{Introduction}

How long does it take for a quantum system to come back to its original state? The Poincar\'e recurrence theorem proves that any classical bounded system will eventually return arbitrarily close to its initial state, and an analogous result applies to quantum systems evolving via a time-independent Hamiltonian with a discrete spectrum \cite{bocchieri1957quantum,peres1982recurrence,hogg1982recurrence}. Interestingly, it was noted by Wallace \cite{wallace2015recurrence} that finite dimensional quantum systems obey a stronger notion of recurrence, in which all states of the system return arbitrarily close to their initial states \emph{at the same time}. This notion of recurrence is referred to as uniform recurrence. 

 In this paper, we will bound the time for such a uniform recurrence of a quantum system, considering unitary evolution in both continuous time and discrete time (e.g. for a quantum walk or quantum cellular automata \cite{Bialynicki-Birula,farrelly2020review}). Bounds on the recurrence time for individual states have previously been obtained for continuous time evolution\footnote{See also `Note Added' at the end of the paper} \cite{bhattacharyya1986estimates,gimeno2017upper,venuti2015recurrence}. For discrete time evolution, expected recurrence times involving repeated measurements of the time evolved state have also been investigated \cite{grunbaum2013recurrence,sinkovicz2015quantized}. While we focus on approximate recurrence in this work, \cite{Anand2026quantumrecurrences} studies exact recurrence for Floquet systems. Another interesting work has linked the framework for equilibration and recurrence together \cite{equilibrium}.

 Our approach incorporates both pure and mixed states, and uses the  trace distance to quantify the strength of recurrence. To bound the recurrence time, we utilise the simultaneous version of Dirichlet approximation theorem for approximation of real numbers by rationals \cite[p.~27]{ApproximationTheorem}. The theorem states that given any $d$ real numbers $\alpha_1,\alpha_2,\dots,\alpha_d$ and any integer $N>1$, there exists integers $q,l_1,l_2,\dots,l_d$ such that $1\leq q \leq N^d$ and
 \be
    \abs{\alpha_i  - l_i/q} \leq \frac{1}{qN} \text{ for all } i. 
    \label{eq:approximationtheoremmain}
 \ee
 Moreover, we include some developments of this approach where we only care about approximating differences between real numbers, which may be of mathematical interest in other situations. 

The paper is structured as follows. In Section \ref{sec:recurrencetime}, we define recurrence time for a quantum system. In Section \ref{sec:weakerbounds}, we use the simultaneous version of Dirichlet's approximation theorem to obtain bounds on recurrence time (for both continuous and discrete time systems). In Section \ref{sec:aprox}, we present a result for approximation of difference of real numbers by difference of rational numbers. These results allow us to obtain better bounds for recurrence time in Section \ref{sec:betterbounds}. We conclude with a discussion in Section \ref{sec:discussion}.

\section{Quantum Recurrence Time}
\label{sec:recurrencetime}

Let us first discuss what we mean by recurrence time for an individual state. Suppose we have some initial state and we let it evolve for some time. We consider some small number $\epsilon$. Now, if for some time $t_{\mathrm{r}}$, the time evolved state is at most $\epsilon$ distance away from the initial state, we might want to call $t_{\mathrm{r}}$ a recurrence time. However, it might be the case that the state never got farther than $\epsilon$ from its initial state. Thus, we will require that the state has to have evolved more than $\epsilon$ away from the initial state at some time before $t_{\mathrm{r}}$ in order to count as a non-trivial recurrence.

In order to define this more formally, consider a Hilbert space $\Hc$, with $D(\Hc)$  the space of all valid density operators on $\Hc$. We will use the trace distance as a measure of distance between two states, which quantifies the probability of determining which state you have via the optimal measurement. Given some $\rho,\sigma\in D(\Hc)$, the trace distance between them is defined as $T(\rho,\sigma) \equiv \tr |\rho-\sigma |/2$.

\begin{definition}[State Recurrence]
Suppose for some initial state $\rho(0) \in \Dc(\Hc)$, we have a family of time-evolved states $\rho(t) \in \Dc(\Hc)$, with $t\in \tau$, where $\tau = \mathbb{R}^{+}$ (continuous time) or $\tau = \mathbb{Z}^{+}$ (discrete time). We say that the state $\rho(0)\in \Dc(\Hc)$ is $\epsilon$-recurrent (with $0 \leq \epsilon < 1$) at time $t \in \tau$ if
\be
    T(\rho(t_{\mathrm{r}}),\rho(0)) \leq \epsilon, 
\ee
and there exists some $t^\prime < t$ (with $t^\prime \in \tau$) such that
\be
    T(\rho(t^\prime),\rho(0)) > \epsilon. 
\ee
\end{definition}

The second condition is important for the recurrence to be non-trivial. (Otherwise, we consider the recurrence to be trivial.) This definition is close to the notion of recurrence considered in \cite{bocchieri1957quantum,gimeno2017upper}. However, we are interested in a stronger, state-independent notion of recurrence. 

\begin{definition}[System Recurrence]
Suppose for some Hilbert space $\Hc$, for every initial state $\rho(0)\in\Dc(\Hc)$, we have a family of time-evolved states $\rho(t)\in D(\Hc)$, with $t\in \tau$, where either $\tau = \mathbb{R}^{+}$ (continuous time) or $\tau = \mathbb{Z}^{+}$ (discrete time). We say that the system is $\epsilon$-recurrent (with $0\leq \epsilon < 1$) at time $t_{\mathrm{r}} \in \tau$ if (i) for all $\rho(0)\in\Hc$,
\be
    T(\rho(t_{\mathrm{r}}),\rho(0)) \leq \epsilon, 
\ee
and (ii) there exists some state $\sigma(0)\in\Dc(\Hc)$ and some time $t^\prime\leq t_{\mathrm{r}}$ (with $t^{\prime} \in \tau$) such that
\be
    T(\sigma(t^\prime),\sigma(0)) > \epsilon. 
\ee
\end{definition}

In other words, for a system to be recurrent at some time, all states must recur trivially or non-trivially with at least one state recurring non-trivially\footnote{Note that we do not require all states to recur non-trivially, as some states might be stationary (e.g. energy eigestates for Hamiltonian evolution).}. This is close to the notion of uniform recurrence in \cite{wallace2015recurrence}. 

We are primarily interested in continuous time Hamiltonian evolution and discrete-time unitary evolution. In these cases, time-evolution preserves the purity of the initial state. Thus, for every initial pure state $\ket{\psi(0)}$, we can write the time-evolved states as $\ket{\psi(t)}$. Then, the following lemma follows from the joint convexity of the trace distance (See Appendix \ref{app:lemmaproof} for proof). 

\begin{lemma}
Suppose we have a quantum system (associated with some Hilbert space $\Hc$) where time-evolution is unitary. The quantum system is $\epsilon$-recurrent at time $t_{\mathrm{r}}$ iff (i) for all pure states $\ket{\psi(0)}\in\Hc$, 
\be
    T(\ket{\psi(t_{\mathrm{r}})},\ket{\psi(0)}) \leq \epsilon,
\ee
and (ii) there exists some pure state $\ket{\varphi(0)}\in\Hc$ and some time $t^\prime < t_{\mathrm{r}}$ such that
\be
    T(\ket{\varphi(t^\prime)},\ket{\varphi(0)}) >\epsilon.
\ee
\label{prop:purity}
\end{lemma}
Here, we have used $T(\ket{\psi}\!,\ket{\phi})$ as a shorthand for $T(\ketbra{\psi}\!,\ketbra{\phi})=\sqrt{1-|\!\braket{\psi}{\phi}\!|^2}$.

Thus, for a system undergoing unitary evolution, we only need to consider pure states when proving system recurrence. More precisely, a unitary quantum system is recurrent at some time iff all pure states recur trivially or non-trivially at that time with at least one pure state recurring non-trivially. 

\section{Bounds on Recurrence Time using Dirichlet's Approximation Theorem}
\label{sec:weakerbounds}

\subsection{Continuous Time Hamiltonian Evolution}

We first consider quantum systems where the evolution is governed by a time-independent Hamiltonian with continuous time. Given a Hamiltonian $H$ (acting on some Hilbert space $\Hc$), for any state $\ket{\psi(0)}\in\Hc$, the family of time-evolved states is given by $\ket{\psi(t)}= e^{-iHt}\ket{\psi(0)}$ with $t\in\tau=\mathbb{R}^{+}$. Here and throughout the paper, we set $\hbar=1$. We can obtain a bound for recurrence times using the Dirichlet's approximation theorem.

\begin{theorem}
 Consider a Hamiltonian $H$ with maximum and minimum energy eigenvalues being $E_{\max}$ and $E_{\min}$ respectively. Suppose that the energy spectrum is discrete and finite with $d \geq 2$ distinct energy eigenvalues. Then, the system is $\epsilon$-recurrent (with $0<\epsilon<1$) at time $t_{\mathrm{r}}$ such that 
        \begin{equation*}
            \quad t_{\mathrm{r}} \leq \frac{2\pi}{E_{\max}-E_{\min}}\left(2\left\lceil \frac{\pi}{\epsilon}\right\rceil\right)^{d-2}.
        \end{equation*}
    \label{thm:hamiltonianfirstbound}
\end{theorem}

\begin{proof}
Suppose we are dealing with a Hilbert space of dimension $D$. Let the eigenvalues and eigenstates of $H$ be $\{E_k\}_{k=1}^{D}$ and $\{\ket{k}\}_{k=1}^{D}$ respectively. As $H$ has $d$ distinct eigenvalues, $D\geq d$. We arrange the eigenvalues such that (i) $E_1=E_{\min}$ and $E_d=E_{\max}$, and (ii) $E_1,E_2,\dots,E_d$ are distinct. Thus, for the degenerate case, for any $k>d$, $E_k=E_j$ for some $j\leq d$.

Let
\be
    \ket{\psi(0)} = \sum_{k=1}^{D} c_k \ket{k}
\ee
be some arbitrary initial state.

We have, after time t,
\be
    \ket{\psi(t)} = \sum_{k=1}^{D} e^{-i E_k t} c_k \ket{k}.
\ee

Now, the trace distance, between the two states is given by,
\begin{align}
    T(\ket{\psi(t)},\ket{\psi(0)})^2 \!&= 1 - |\braket{\psi(t)}{\psi(0)}|^2 \nonumber\\
    &\!\!\!\!\!\!\!\!= 1 - \sum_{j,k}\abs{c_j}^2\abs{c_k}^2 e^{-i(E_j\!-\!E_k)t} \nonumber\\
    &\!\!\!\!\!\!\!\!= 1 - \sum_{j,k} \abs{c_j}^2\abs{c_k}^2 \cos{[(E_j\!-\!E_k)t]} \nonumber\\
    &\!\!\!\!\!\!\!\!=  \sum_{j,k} \abs{c_j}^2\abs{c_k}^2 (1- \cos{[(E_j\!-\!E_k)t]}). 
    \label{eq:tracedistance}
\end{align}

Now, suppose we can find a time $t_{\mathrm{r}}$ such that
\be
    \abs{(E_j-E_k)t_{\mathrm{r}} - 2\pi (l_{j}-l_{k})} \leq 2 \epsilon
    \label{eq:energyapprox}
\ee
for some integers $l_{k}$, and $j,k\in\{1,2,\dots,d\}$. (It automatically follows that the above equation holds for $j,k=1,2,\dots,D$.) Defining $\alpha_k = E_k t_{\mathrm{r}} - 2\pi l_k$, we can rewrite the above equation as 
\be
    \abs{(E_j-E_k)t_{\mathrm{r}} - 2\pi (l_{j}-l_{k})} = \abs{(\alpha_j-\alpha_k)} \leq 2 \epsilon.
    \label{eq:energyapproxderived}
\ee

Then, it follows from Eq.~\eqref{eq:tracedistance} that
\begin{align}
    T(\ket{\psi(t_{\mathrm{r}})},\ket{\psi(0)})^2
    &= \sum_{j,k} \abs{c_j}^2\abs{c_k}^2 (1- \cos{[(\alpha_j\!-\!\alpha_k)t]})
    \nonumber \\
    &\leq   \frac{1}{2}\sum_j \sum_k \abs{c_j}^2\abs{c_k}^2 (\alpha_j-\alpha_k)^2
    \label{eq:tracedistanceinterim}
\end{align}
where we have used the periodicity of cosine in the first line and the inequality $1-\cos(x) \leq x^2/2$ in the second line.

If we let $A$ be a random variable with outcomes $\alpha_j$ and the respective probabilities being $p_j= \abs{c_j}^2$, we get

\begin{align}
    \frac{1}{2}\sum_{j,k} p_jp_k (\alpha_j-\alpha_k)^2 &= \frac{1}{2}\sum_{j,k} p_jp_k (\alpha_j^2+\alpha_k^2\!-\!2\alpha_j\alpha_k)
    \nonumber\\
    &= \langle A^2 \rangle - \langle A \rangle^2 = \text{Var}(A).
\end{align}

Thus, we have
\begin{align}
    T(\ket{\psi(t_{\mathrm{r}})},\ket{\psi(0)})^2
    \leq\text{Var}(A)
    \label{eq:tracedistancevariance}.
\end{align}

Then, using Popoviciu's inequality on variances, we have
\be
    \text{Var}(A) \leq \frac{(\max_k{\alpha_k}-\min_k{\alpha_k})^2}{4} \leq \frac{(2\epsilon)^2}{4} = \epsilon^2.
\ee

Finally, we have $T(\ket{\psi(t_{\mathrm{r}})},\ket{\psi(0)})\leq \epsilon$. Now, we have reduced the problem to proving Eq.~\eqref{eq:energyapprox} with $t_{\mathrm{r}}$ following the appropriate bounds.

We can rewrite Eq.~\eqref{eq:energyapprox} as
\be
     \abs{(E_j-E_k)\frac{t_0}{2\pi}q - (l_{j}-l_{k})} \leq  \frac{\epsilon}{\pi}.
     \label{eq:energyDiffMain}
\ee
where
\be
    t_0 = \frac{2\pi}{E_{d}-E_{1}}= \frac{2\pi}{E_{\text{max}}-E_{\text{min}}}
\ee
and we let $t_{\mathrm{r}} = qt_0$.

With $N=\lceil\pi/\epsilon\rceil$, using the simultaneous version of Dirichlet's approximation theorem (Eq.~\eqref{eq:approximationtheoremmain}), there exists some integer $q$ satisfying $1\leq q \leq (2N)^{d-2}$ and some integers $l_k$, such that
\be
      \abs{(E_k-E_1)\frac{t_0}{2\pi}q-l_{k}} < \frac{1}{2N}
      \label{eq:dirichletBoundEnergy}
\ee
where $k\in\{2,\dots,d-1\}$. Now, note that Eq.~\eqref{eq:dirichletBoundEnergy} also holds true for (i) $k=d$ with $l_d = q$ and (ii) $k=1$ with $l_1 = 0$ for any $q$, and (iii) $k>d$ because of the way we have arranged the energy eigenvalues. Thus, we can say that there exists some integers $l_k$  and $q$ with $1\leq q \leq (2N)^{d-2}$ such that Eq.~\eqref{eq:dirichletBoundEnergy} holds for all $k$.

Now, for any $j,k$, consider
\begin{align}
    \abs{(E_j-E_k)\frac{t_0}{2\pi}q-(l_j-l_k)} \leq& 
    \nonumber \\ \abs{(E_j-E_1)\frac{t_0}{2\pi}q-l_{j}}+&\abs{(E_k-E_1)\frac{t_0}{2\pi}q-l_{k}} < \frac{1}{N}.
\end{align}
Thus, we have obtained Eq.~\eqref{eq:energyDiffMain}. 

Then, noting that $t_{\mathrm{r}} = q t_0$, we get
\be
 t_0 \leq t_{\mathrm{r}} \leq t_0 \left(2\left\lceil \frac{\pi}{\epsilon}\right\rceil\right)^{d-2}.
 \label{eq:range}
\ee

Now, we need to show that there exists at least one state that has non-trivial recurrence. 
Consider the state 
\be
    \ket{\varphi(0)} = \frac{1}{\sqrt{2}}(\ket{1}+\ket{d}).
\ee
Then, we have, for any time $t$,
\begin{align}
    T(\ket{\varphi(t)},\ket{\varphi(0)})^2 = \sin^2((E_d-E_1)t/2).
\end{align}

Now, we can always find a $t^\prime$ with $0\leq t^\prime(E_d-E_1)/2 < \pi$ such that $\sin^2((E_d-E_1)t^\prime/2)> \epsilon^2$ (for any $\epsilon<1$). Thus, we have, 
\begin{align}
    T(\ket{\varphi(t^\prime)},\ket{\varphi(0)}) = \abs{\sin((E_d-E_1)t^\prime/2)} > \epsilon.
\end{align}

Note that $t^\prime < t_0 \leq t_{\mathrm{r}}$. This concludes our proof.
\end{proof}

\subsection{Discrete Time Unitary Evolution}

We shall now consider the case when time is discrete, and the system evolves in each time step via a unitary $U$. For any initial state $\ket{\psi(0)}$, we can write the time-evolved states as $\ket{\psi(n)}=U^{n}\ket{\psi(0)}$ where $n\in\tau = \Zb^{+}$.

Now, similarly to the continuous time case, we can get a bound for recurrence time when considering discrete time unitary evolution using Dirichlet's approximation theorem. Before we present the bound, we define distance on a circle via
\be
    R(\theta_1,\theta_2) = \min (\abs{\theta_1-\theta_2},2\pi-\abs{\theta_1-\theta_2})
\ee
where $\theta_1,\theta_2\in[0,2\pi)$.

\begin{theorem}
Suppose we have a unitary $U$, having a discrete and finite spectrum, with $d \geq 2 $ distinct eigenvalues $\{e^{i\phi_k}\}$ where $0\leq \phi_k<2\pi$. Then, for any $0<\epsilon\leq1/2$, the quantum system is $\epsilon$-recurrent at some time $m_{\mathrm{r}}\in\mathbb{N}$ such that
\begin{align*}
        m_{\mathrm{r}} &\leq \frac{\pi}{\max_{jk}R(\phi_j,\phi_k)}\left(2\left\lceil \frac{\pi}{\epsilon}\right\rceil\right)^{d-1}.
\end{align*}
\end{theorem}
\begin{proof}
    Suppose that the dimension of the system is $D$. Then, suppose $U$ has eigenvalues $\{e^{i\phi_k}\}_{k=1}^{D}$ and eigenstates $\{\ket{k}\}_{k=1}^{D}$, with $\phi_k\in[0,2\pi)$. As $U$ has $d$ distinct eigenvalues, we have $D\geq d$. Moreover, we can arrange the eigenvalues such that $\phi_1,\phi_2,\dots,\phi_d$ are all distinct.

    Suppose we have an arbitrary initial state $\ket{\psi(0)}=\sum_k c_k \ket{k}$. Then, we have $\ket{\psi(m)}=\sum_k c_k e^{i\phi_k m}\ket{k}$ for any $m\in\mathbb{Z}^{+}$. Now, similar to the proof of theorem \ref{thm:hamiltonianfirstbound}, we have 
    \be
       T(\ket{\psi(0)},\ket{\psi(m)}) = \sum_{j,k} \abs{c_j}^2\abs{c_k}^2 (1- \cos{[(\phi_j\!-\!\phi_k)m]})
       \label{eq:discretedifference}
    \ee
    for any time $m$.
    
     Suppose we can find a time $m_{\mathrm{r}}$ such that
    \be
        \abs{(\phi_j-\phi_k)m_{\mathrm{r}} - 2\pi (l_{j}-l_{k})} \leq 2 \epsilon
        \label{eq:energyapproxdiscrete}
    \ee
    for all $j,k$ and where $l_j$ are some integers. Then, following similar steps as before, we get
    \be
         T(\ket{\psi(0)},\ket{\psi(m_{\mathrm{r}})}) \leq \epsilon.
    \ee
    
     Let $m_0$ be an arbitrary positive integer. With $N=\lceil\pi/\epsilon\rceil$, using Dirichlet's approximation theorem (Eq.~\eqref{eq:approximationtheoremmain}), there exists some integer $q$ satisfying $1\leq q \leq (2N)^{d-2}$ and some integers $l_k$ such that
    \be
          \abs{(\phi_k-\phi_1)\frac{m_0}{2\pi}q-l_{k}} < \frac{1}{2N} \label{eq:dirichletBoundEnergyDiscrete}
    \ee
    where $k\in\{1,2,\dots,d-1\}$. Now, note that Eq.~\eqref{eq:dirichletBoundEnergyDiscrete} also holds true for (i) $k=1$ with $l_d = 0$ and any $q$, and (ii) $k>d$ because of the way we have arranged the eigenvalues. 
    Now, for any $j,k$, we have
    \begin{align}
        \abs{(\phi_j-\phi_k)\frac{m_0}{2\pi}q-(l_j-l_k)} \leq& 
        \nonumber \\ 
        \abs{(\phi_j-\phi_1)\frac{m_0}{2\pi}q-l_{j}}+&\abs{(\phi_k-\phi_1)\frac{m_0}{2\pi}q-l_{k}} < \frac{1}{N}.
    \end{align}
    Thus, we have obtained Eq.~\eqref{eq:energyapproxdiscrete} with $m_{\mathrm{r}} =  m_0q$. 
    
    To summarise, we have shown that for any positive integer $m_0$, we can find an $m_{\mathrm{r}}$ such that 
     \be
        m_0 \leq m_{\mathrm{r}} \leq m_0\left(2\left\lceil \frac{\pi}{\epsilon}\right\rceil\right)^{d-1}
        \label{eq:generalm0}
     \ee
     and all states are (trivially or non-trivially) recurrent at $m_{\mathrm{r}}$. 
     
     We now need to pick an $m_0$ such that we have at least one state having non-trivial recurrence. Consider the state
    \be
        \ket{\varphi(0)}=\frac{1}{\sqrt{2}}(\ket{1}+\ket{2})
    \ee
    where $\ket{1}$ and $\ket{2}$ are any arbitrary eigenstates of $U$ with distinct eigenvalues.
    After some time $m\in\Zb$, we have,
    \be
        \ket{\varphi(m)}=\frac{1}{\sqrt{2}}(e^{im\phi_1}\ket{1}+e^{im\phi_2}\ket{2}).
    \ee
    We also have,
    \begin{align}
        T(\ket{\varphi(m)},\ket{\varphi(0)})^2 
        &= \sin^2((\phi_2-\phi_1)m/2)
        \nonumber\\
        &=\sin^2(R(\phi_1,\phi_2)m/2).
        \label{eq:distancefinal}
    \end{align}
    
    Now, note that there always exists an integer $m^\prime$ such that $m^\prime R(\phi_1,\phi_2)\in(\pi/3,\pi]$ for all $0<R(\phi_1,\phi_2)\leq\pi$. Thus, for such an $m^\prime$, we have
    \begin{align}
         T(\ket{\varphi(m^\prime)},\ket{\varphi(0)})^2&= 
         \sin^2(R(\phi_1,\phi_2)m^\prime/2) 
         \nonumber\\
         &> \sin^2(\pi/6) =1/4 \geq \epsilon^2,
    \end{align}
    and
    \be
        m^\prime < \frac{\pi}{R(\phi_1,\phi_2)}.
    \ee
    
    Now, we pick $\phi_1$ and $\phi_2$ such that $R(\phi_1,\phi_2)=\max_{jk}R(\phi_j,\phi_k)$. Thus, we have
    \be
        m^\prime < \frac{\pi}{\max_{jk}R(\phi_j,\phi_k)}.
    \ee
    Finally, picking $m_0 = {\pi}/{\max_{jk}R(\phi_j,\phi_k)}$ and substituting in Eq.~\eqref{eq:generalm0}, we obtain the bound given in the theorem.
\end{proof}

\section{Approximation of Difference of Reals by Rationals}
\label{sec:aprox}

Now, note that the approach above concerns approximating individual energies and phases by rational approximations (via Eq.~\eqref{eq:dirichletBoundEnergy} and Eq.~\eqref{eq:dirichletBoundEnergyDiscrete}). However, the recurrence behaviour really depends on energy differences given by
\be
    \abs{(\phi_j-\phi_k)m_{\mathrm{r}} - 2\pi (l_{j}-l_{k})} \leq 2 \epsilon.   \tag{\ref{eq:energyapproxdiscrete}}
\ee
Here, we show that we can obtain improved bounds by optimizing the rational approximations of such differences more directly. 

To make the situation a bit more general, for any set of real numbers  $\{\alpha_i\}_{i=1}^d$ and any natural number $N$, we want to find integers $\{l_{i}\}_{i=1}^{d}$ and an integer $q$ such that\footnote{Note that this is a strict inequality, and hence, a slightly stronger condition than Eq.~\eqref{eq:energyapproxdiscrete}.} 
\begin{equation}
    \abs{(\alpha_i-\alpha_j)-(l_{i}-l_{j})/q} < 1/Nq.
    \label{eq:approximationtheorem}
\end{equation}
Furthermore, we want to find the smallest integer $M$ (dependent on $N$ and $d$) such that $1 \leq q \leq M$. Using the simultaneous version of Dirichlet's approximation theorem and following a similar approach to that used earlier, we can obtain the bound $M=(2N)^{d-1}$.

How can we do better? Note that one approach to proving the Dirichlet's approximation theorem involves the pigeonhole principle\footnote{If you have $n+1$ items places in $n$ pigeonholes, one of the pigeonholes must contain at least two items.}. The approach considered earlier, focussing on approximating each $\alpha_i$ individually, corresponds to dividing the parameter space $[0,1)^{d-1}$ into pigeonholes which are small hypercubes $[\,0,1/(2N)\,)^{d-1}$. However, when optimized to approximate differences between the rational approximations, it is helpful to use pigeonholes of different shapes, as shown in figure \ref{fig:shapeoftiles}. It turns out that $M$ corresponds to the number of pigeonholes, and by making each individual pigeonhole larger, $M$ is reduced. This leads to improved bounds for situations in which $N$ is sufficiently large (corresponding to good approximation), as given below. These might be useful in other situations in which the differences between approximated quantities play a central role. 

\begin{figure}
	\centering
	\begin{subfigure}{0.45\linewidth}
        \centering
		\includegraphics[width=\linewidth]{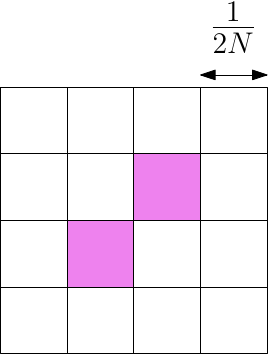}
		\caption{Region $\mathcal{P}$}
        \label{fig:regionP}
	\end{subfigure}
    \hfill
    \begin{subfigure}{0.45\linewidth}
        \centering
		\includegraphics[width=\linewidth]{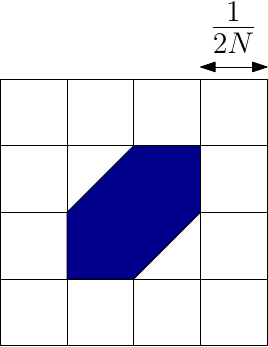}
		\caption{Region $\mathcal{T}$}
        \label{fig:regionT}
	\end{subfigure}
 \caption{The pigeonholes used to cover $[0,1)^{d-1}$ when $d=3$. For the simple application of the Dirichlet approximation theorem, the corresponding pigeonhole would be a square of side length $1/(2N)$.}
 \label{fig:shapeoftiles}
\end{figure}

\begin{proposition}
    For any $d\geq 2$ distinct real numbers $\alpha_1,\alpha_2,\dots,\alpha_d$ and any positive integer $N$, there exists some integers $l_1,l_2,\dots,l_d$ and $q\geq 1$ with 
    \begin{align*}
            q \leq (2N)^{d-1},
            \nonumber \\
            q \leq N(2N+1)^{d-2},
            \nonumber \\
            q \leq \frac{(2N+2)^{d-1}}{d},
        \end{align*}
    such that
        \begin{equation*}
            \abs{(\alpha_i-\alpha_j)-(l_{i}-l_{j})/q} < 1/Nq
        \end{equation*}
    for all $i,j$. 
    \label{prop:approxtheorem}
\end{proposition}
The proof of this proposition is presented in Appendix \ref{app:proofofprop}. Finally, if we are also given that for some given $n$ and $m$, $\alpha_n-\alpha_m$ is an integer, then we can reduce $d$ by one in the upper bounds for $q$. Then, this modified result holds for $d\geq 3$.

\section{Better Bounds for Recurrence Time }
\label{sec:betterbounds}
Now, let us go back to the proof of Theorem \ref{thm:hamiltonianfirstbound}. We wished to find integers $q$ and $l_k$ such that the following equation holds for all $j,k$.
\be
      \abs{(E_j-E_k)\frac{t_0}{2\pi}q - (l_{j}-l_{k})} \leq  \frac{\epsilon}{\pi}.
     \tag{\ref{eq:energyDiffMain}}
\ee
Now, if we let $\alpha_k=E_k {t_0}/{2\pi}$, $N=\lceil \pi/\epsilon\rceil$, and note that $\alpha_{\max}-\alpha_{\min}=(E_{\max}-E_{\min})t_0/2\pi = 1 \in \Zb$, we can use the modified version of proposition \ref{prop:approxtheorem} to obtain the above equation with an upper bound for $q$. Then, proceeding as before, we can prove the following theorem.

\begin{theorem}
    Consider a Hamiltonian $H$ with the maximum and minimum energy eigenvalues being $E_{\max}$ and $E_{\min}$ respectively. Suppose that the energy spectrum is discrete and finite with number of distinct energy eigenvalues being $d\geq 3$. Then, the system is $\epsilon$-recurrent (with $0<\epsilon<1$) at some time $t_{\mathrm{r}}$ such that
    \begin{align*}
         t_{\mathrm{r}} &\leq \frac{2\pi}{E_{\max}-E_{\min}}\left\lceil \frac{\pi}{\epsilon}\right\rceil\left(2\left\lceil \frac{\pi}{\epsilon}\right\rceil+1\right)^{d-3},
         \\
          t_{\mathrm{r}} &\leq \frac{2\pi}{E_{\max}-E_{\min}}\frac{1}{d-1}\left(2\left\lceil \frac{\pi}{\epsilon}\right\rceil+2\right)^{d-2}.
    \end{align*}
    Note that the second bound also holds true for $d=2$. 
\end{theorem}

Similarly, we can use proposition \ref{prop:approxtheorem} to prove the following theorem for discrete time evolution.

\begin{theorem}
Suppose we have a unitary $U$, having a discrete and finite spectrum, with $d \geq 2 $ distinct eigenvalues $\{e^{i\phi_k}\}$ where $0\leq \phi_k<2\pi$. Then, for any $0<\epsilon\leq1/2$, the quantum system is $\epsilon$-recurrent at some time $m_{\mathrm{r}}\in\mathbb{N}$ such that
 \begin{subequations}
    \begin{align}
           m_{\mathrm{r}} &\leq \frac{\pi}{\max_{jk}R(\phi_j,\phi_k)}\left\lceil \frac{\pi}{\epsilon}\right\rceil\left(2\left\lceil \frac{\pi}
           {\epsilon}\right\rceil+1\right)^{d-2}, 
            \\
            m_{\mathrm{r}} &\leq \frac{\pi}{\max_{jk}R(\phi_j,\phi_k)}\frac{1}{d}\left(2\left\lceil \frac{\pi}{\epsilon}\right\rceil+2\right)^{d-1}.
    \end{align}
 \end{subequations}
\end{theorem}

\section{Discussion}
\label{sec:discussion}
For any unitarily evolving finite dimensional quantum system, we can find a  time at which all states are arbitrarily close to their initial configuration, and at least one state has deviated significantly over this time interval. We have obtained bounds on the recurrence time for this notion of system recurrence, in both continuous and discrete time. These scale approximately like $(\frac{1}{\epsilon})^{d-2}$ or 
$(\frac{1}{\epsilon})^{d-1}$ respectively, where $d$ is the number of distinct energies for the system, and $\epsilon$ is the accuracy of recurrence. Our approach also introduced some mathematical techniques to approximate the differences between real numbers via rationals, which may be of interest in other circumstances. 

Although we would not expect similar recurrences to  occur in general for infinite dimensional systems (even when their spectra are discrete), we would expect that similar results could be obtained by placing a restriction on initial states and Hamiltonians (e.g. a finite number of energy eigenvalues in any finite energy interval, and a bound on the ground state energy of the Hamiltonian and the expected energy of the initial state), as in \cite{wallace2015recurrence}.

It would also be interesting to consider the recurrence of open quantum systems or general quantum channels (e.g. systems evolving via Lindblad equations or non-unitary CPTP maps) in a similar context. While we would expect that such evolutions would not yield recurrences for all $\epsilon$ due to contractivity, it would be good to obtain bounds on the recurrence behaviour in such cases. 

\section*{Note Added}
\label{sec:note}
After the completion of this work, we became aware of a recent independent work of Kotowski and Oszmaniec \cite{kotowski2026tight}, which has related results. The main focus of their paper is on individual state recurrence, where they obtain many interesting results. In the supplementary information (Theorem S1), they also obtain a similar result to our Theorem \ref{thm:hamiltonianfirstbound}. 


\begin{acknowledgments}
We are grateful to Alex Essery for his insightful discussions on the tiling proof.
\end{acknowledgments}

\appendix

\section{Proof of Lemma \ref{prop:purity}}
\label{app:lemmaproof}

Given some arbitrary $t \in \tau$, it is sufficient to show that (i) $T(\rho(t),\rho(0))\leq\epsilon$ for all $\rho(0)$ iff $T(\ket{\psi(t)},\ket{\psi(t)})\leq\epsilon$ for all $\ket{\psi(0)}$, and (ii) $T(\sigma(t),\sigma(0))>\epsilon$ for some $\sigma(0)$ iff $T(\ket{\varphi(t)},\ket{\varphi(t)})>\epsilon$ for some $\ket{\varphi(0)}$.

For (i), the `only if' direction is trivial. For the `if' direction, consider some arbitrary density operator $\rho(0)=\sum_i p_i \ketbra{\psi_i(0)}$. For unitary time evolution, we can write $\rho(t)=\sum_i p_i \ketbra{\psi_i(t)}$. Then, we have
\be
    T(\rho(t),\rho(0)) \leq \sum_i p_i T (\ket{\psi_i(t)},\ket{\psi_i(0)}) \leq \epsilon
\ee
where the first inequality follows from joint convexity of the trace distance and the second follows from assumption.

For (ii), the `if' direction is trivial. For the
`only if' direction, we write $\sigma(0)=\sum_i p_i \ketbra{\varphi_i(0)}$, and $\sigma(t)=\sum_i p_i \ketbra{\varphi_i(t)}$ (unitary evolution). Then, we have
\be
\sum_i p_i T (\ket{\varphi_i(t)},\ket{\varphi_i(0)}) \geq  T(\sigma(t),\sigma(0)) > \epsilon
\ee
where, again, the first inequality follows from joint convexity of the trace distance and the second from assumption. 

Noting that the first expression in the above equation is an average, we can deduce that there must exist some $j$ such that $ T (\ket{\varphi_j(t)},\ket{\varphi_j(0)})>\epsilon$.

\section{Approximation Theorem by Tiling (Proof of Proposition \ref{prop:approxtheorem})}
\label{app:proofofprop}
Suppose we are given some real numbers $\{\alpha_i\}_{i=1}^d$ and some natural number $N$. We suppose that $\alpha_1$ is the smallest. Let $M$ be some arbitrary whole number. 
Note that, we can write, for $m=0,1,\dots M$ and $i=2,3,\dots d$,
\be
     (\alpha_i-\alpha_1)m = x^{(m)}_{i}+k^{(m)}_{i}
     \label{eq:decomposition}
\ee
where $k^{(m)}_{i}$ are integers and $0\leq x^{(m)}_{i}< 1$.

Now, for each $m$, we define a vector $\mathbf{x}^{(m)}$ by letting its components be $x^{(m)}_{i}$. Also, consider the region $\mathcal{C}=[0,1)^{(d-1)}$. Then, $\mathbf{x}^{(m)}\in\mathcal{C}=[0,1)^{(d-1)}$ for all $m$.

Now, suppose there exists a partition $P$ for some set $\mathcal{C}^\prime \supseteq \mathcal{C}$, such that for every cell $\mathcal{R}\in P$, and for any two points $\mathbf{z},\mathbf{y}\in\mathcal{R}$,
\begin{subequations}
\begin{align}
    \abs{(z_i - y_i)-(z_j - y_j)} &< \frac{1}{N},
    \label{eq:tiledistance}
    \\
    \abs{z_i-y_i} &< \frac{1}{N}.
    \label{eq:tiledistanceindividual}
\end{align}\label{eq:tiledistancemain}
\end{subequations}
for all $i,j$. Then, we pick $M = \abs{P}$. 

Now, we have $M+1$ vectors $\mathbf{x}^{(m)}$ as $m=0,1,2,\dots,M$, but the number of tiles in the partition are $M$. Hence, using the pigeonhole principle, we must have some $\mathbf{x}^{(m)}$
and $\mathbf{x}^{(n)}$, with $n>m$, lying in the same cell. Then, from Eq.~\eqref{eq:tiledistancemain}, it follows that

\begin{subequations}
\begin{align}
    \abs{(x^{(n)}_{i} - x^{(m)}_{i})-(x^{(n)}_{j} - x^{(m)}_{j})} < \frac{1}{N},
    \label{eq:tesselationbound}
    \\
    \abs{x^{(n)}_{i} - x^{(m)}_{i}} < \frac{1}{N}.
    \label{eq:tesselationindividual}
\end{align}\label{eq:tesselationboundmain}
\end{subequations}

Note that $1\leq n - m \leq M$. Take $q= n - m$. For $i,j = 2,3,\dots d$, consider
\begin{align}
    (\alpha_i-\alpha_j)q &=  (\alpha_i-\alpha_1)q-(\alpha_j-\alpha_1)q
    \nonumber\\
    &= (k^{(n)}_{i} - k^{(m)}_{i})-(k^{(n)}_{j} - k^{(m)}_{j}) 
    \nonumber\\
    &\quad+ (x^{(n)}_{i} - x^{(m)}_{i})-(x^{(n)}_{j} - x^{(m)}_{j}).
\end{align}
We also have, for $i=2,3,\dots d$,
\begin{align}
    (\alpha_i-\alpha_1)q 
    &= (k^{(n)}_{i} - k^{(m)}_{i}) + (x^{(n)}_{i} - x^{(m)}_{i}).
\end{align}

Then, with $l_{i}= k^{(m)}_{i} - k^{(n)}_{i}$ for $i \in\{2,3,\dots d\}$ and $l_1 = 0$, we obtain,
\begin{equation}
    \abs{(\alpha_i-\alpha_j)q-(l_{i}-l_{j})} < 1/N
\end{equation}
for all $i,j$ with $1\leq q \leq M$.

Now, note that if we have a region $\mathcal{R}$ such any two points in $\mathcal{R}$ satisfy Eqs.~\eqref{eq:tesselationboundmain}, all translations of $\mathcal{R}$ must also satisfy the same. Hence, if we can construct a tessellation using $\mathcal{R}$ as a tile, we obtain the partition we desired, with $M$ being the number of tiles required to cover the region $\mathcal{C}$. 

Consider a hypercubic tile of side length $1/(2N)$. Now, any two points in this hypercube satisfy the requisite conditions (Eqs.~\eqref{eq:tesselationboundmain}). Noting that it would take $M=(2N)^{d-1}$ such tiles to tessellate $\mathcal{C}$, we have obtained our first bound. 

Now, consider another region,
\be
    \mathcal{P} = \left[0,\tfrac{1}{2N}\right)^{d-1}\cup\left[\tfrac{1}{2N},\tfrac{1}{N}\right)^{d-1}
    \label{eq:twocubedef}
\ee
which is a union of two hypercubes (figure \ref{fig:regionP}). In Appendix \ref{app:twocubestiling}, we show that we can use this region to tile some $\mathcal{C}^\prime \supseteq \mathcal{C} = [0,1)^{d-1}$ using $N(2N+1)^{d-2}$ tiles. This gives us the bound $M=N(2N+1)^{d-2}$. For large $N$, we have $M\approx (2N)^{d-2}/2$, which is a small improvement over the previous bound.

Finally, we consider the following region (figure \ref{fig:regionT}).
\begin{align}
    \mathcal{T} = \Big\{\mathbf{x}\in \mathbb{R}^{d-1} \,|\, -\tfrac{1}{2N} \leq x_i < \tfrac{1}{2N}
    \forall i \nonumber\\
    \text{ and} -\!\tfrac{1}{2N} \leq  (x_j-x_i) < \tfrac{1}{2N} \,\,\forall j>i \Big\}.
    \label{eq:polytopedef}
\end{align}
In Appendix \ref{app:convexhulltiling}, we show that $\mathcal{T}$ tiles $R^{d-1}$. Now, one can reason that the smallest partition of $\mathcal{C}$ formed by tiling it with $\mathcal{T}$ would have to be a subset of $[-1/(2N),1+1/(2N)]^{d-1}$ (as some tiles would be protruding out of $\mathcal{C}$). Now, noting that the volume of each tile is $d/(2N)^{d-1}$ (Appendix \ref{app:volume}), we claim that the number of tiles have to be at most 
\be
    {\frac{(2N+2)^{d-1}}{d}}
\ee
which is just the ratio of the hypervolume of $[-1/(2N),1+1/(2N)]^{d-1}$ to the hypervolume of $\mathcal{T}$. For $N\gg1$, we have $M \approx (2N)^{d-1}/d$.

\section{Tiling via Union of Two Hypercubes}
\label{app:twocubestiling}
In this section, we shall show that $\mathcal{P}$ can be used to construct a tiling covering $\mathcal{C}$. We shall first consider the case when $d \geq 3$. Recall,
\be
    \mathcal{P} = \left[0,\tfrac{1}{2N}\right)^{d-1}\cup\left[\tfrac{1}{2N},\tfrac{1}{N}\right)^{d-1}
    \tag{\ref{eq:twocubedef}}.
\ee
It can be verified that any two points in $\mathcal{P}$ satisfy Eq.~\eqref{eq:tesselationboundmain}. We shall construct a tessellation of a region $\mathcal{C}^\prime \supseteq [0,1)^{d-1}$ by tiling $\mathcal{P}$. 

Consider the region 
\begin{align}
\mathcal{B}&=\left[-\tfrac{1}{2N},1\right)^{d-2}\times[0,\tfrac{1}{2N}) 
\nonumber\\
&\quad\cup \left[0,1+\tfrac{1}{2N}\right)^{d-2}\times \left[\tfrac{1}{2N},\tfrac{1}{N}\right).
\end{align}

Now, we can tessellate this region $\mathcal{B}$ by translating $\mathcal{P}$ by the following set of vectors 
\be
    A = \left\{\frac{1}{2N}\sum_{k=1}^{d-2}n_k \mathbf{e}_k | n_k \in \{-1,0,1,\dots,2N-1\}\right\}
\ee
where $\mathbf{e}_{j}$ is the unit vector in the $j^{\text{th}}$ direction. Note that $\abs{A}=(2N+1)^{d-2}$. 

Now, suppose we have some tiling formed by translating $\mathcal{B}$ by the following set of vectors 
\be
    B =\left\{\frac{k}{N}\,\mathbf{e}_{d-1}| k = 0,1,\dots,N-1\right\}.
\ee
Here, the number of tiles is $\abs{B}=N$. Let the union of all the tiles be the region $\mathcal{C}'$. One can verify that $[0,1)^{d-1}\subset \mathcal{C}'$. 

To summarise, we can tile $\mathcal{B}$ by translating $\mathcal{P}$ by all the vectors in $A$, and we can tile $\mathcal{C}^\prime$  by translating $\mathcal{B}$ by all the vectors in $B$. It follows that we can tessellate $\mathcal{C}^\prime$ by considering translations of $\mathcal{P}$, with the number of tiles being $\abs{A}\abs{B}=N(2N+1)^{d-2}$. 

Finally, for $d=2$, we can simply partition $(0,1]$ into $(k/N,(k+1)/N]$ with $k=0,1,\dots,N-1$, and obtain the bound $N$. Thus, the above result also works for $d=2$.

\section{Tiling via Convex Hull of Two Hypercubes}
\label{app:convexhulltiling}
Recall that we had defined the region $\mathcal{T}$ as
\begin{align}
    \mathcal{T} = \Big\{\mathbf{x}\in \mathbb{R}^{d-1} \,|\, -\tfrac{1}{2N} \leq x_i < \tfrac{1}{2N}
    \forall i \nonumber\\
    \text{ and} -\!\tfrac{1}{2N} \leq  (x_j-x_i) < \tfrac{1}{2N} \,\,\forall j>i \Big\}.
    \tag{\ref{eq:polytopedef}}
\end{align}

 We shall prove that $\mathcal{T}$ tiles $\mathbb{R}^{d-1}$ with the translation vectors being $\mathbf{v}^{(i)}$,
\begin{align}
    \mathbf{v}^{(i)} &= (1,1,\dots,1,\overbrace{2}^{i^\text{th} \text {component}},1,\dots,1).
\end{align}
Let $\mathbf{x}$ be any point in $\mathbb{R}^{d-1}$. We wish to show that there exist unique integers $\{n_i\}$ such that $\mathbf{x}-\sum_i \tfrac{n_i}{2N} \mathbf{v}^{(i)}\in \mathcal{T}$. 

Now, $\mathbf{x}-\sum_i \tfrac{n_i}{2N} \mathbf{v}^{(i)}\in \mathcal{T}$ iff
\be
    2N\mathbf{x}= \sum_i (n_i + \beta_i) \mathbf{v}^{(i)},
    \label{eq:decompx}
\ee
and
\begin{subequations}
    \begin{align}
    (\mathrm{I}) & \quad \beta_{\max} + \sum_k \beta_k < 1, \\
    (\mathrm{II}) & \quad \beta_{\min} + \sum_k \beta_k \geq  -1, \\
    (\mathrm{III}) & -1 \leq \beta_j - \beta_i < 1 \text{ for all } j > i
    \end{align}
\end{subequations}
for some real numbers $\{\beta_i\}$. We shall show existence of such $\{n_i\}$ first.

First, we shall show that we can modify condition (III) to be a non-strict inequality on both sides. In particular, for existence, we wish to replace the above conditions by the following.
\begin{subequations}
 \begin{align}
    (\mathrm{I}) & \quad \beta_{\max} + \sum_k \beta_k < 1, \\
    (\mathrm{II}) & \quad \beta_{\min} + \sum_k \beta_k \geq  -1, \\
    (\mathrm{III}') & \quad \beta_{\max} - \beta_{\min} \leq 1 
\end{align}
\label{eq:newconditions}
\end{subequations}

To this goal, we consider the case when the new set of conditions are satisfied but the old set of conditions might not be. Suppose, for some $\mathbf{x}$ we are given some integers $\{n_i\}$ and $\{\beta_i\}$ such that Eq.~\eqref{eq:decompx} is true and conditions (I) and (II) are satisfied. Moreover, also suppose that $\beta_{\max}-\beta_{\min} = 1$.
Let $J,I \subset \{1,2,\dots,d-1\}$ such that for all $i\in I$ and $j\in J$, $\beta_j=\beta_{\max}$ and $\beta_i=\beta_{\min}$. Now, if $j < i$ for all $j \in J$ and $i \in I$, the previous conditions are satisfied. 
Consider the case when the above is not true. That is, there exists some $i\in I$ and $j \in J$ such that $j>i$. Clearly, condition (III) is not satisfied. Then, for such $i,j$, we take $\beta_j \to \beta_j-1$ and $\beta_i \to \beta_i+1$, along with $n_j \to n_j + 1$ and $n_i \to n_i - 1$. One can verify that all the conditions are now satisfied. Thus, to show existence, it is sufficient to find $\{\beta_k\}$ and $\{n_k\}$ such that  Eq.~\eqref{eq:decompx} is true and the new set of conditions (Eqs.~\eqref{eq:newconditions}) are satisfied.

Now, due to the linear independence of $\mathbf{v}^{(i)}$, we can write every $\mathbf{x}$ as
\be
    2N\mathbf{x}= \sum_i (n_i + \beta_i) \mathbf{v}^{(i)}
\ee
where $-1/2<\beta\leq 1/2$ and $n_i$ is an integer. Without loss of generality, we suppose that $\beta$ are arranged in ascending order, i.e., $\beta_1\leq\beta_2\leq \dots \leq\beta_{d-1}$. Clearly, condition (III$^\prime$) is satisfied. Now, we have three cases: (i) both (I) and (II) are also satisfied, (ii) (I) is not satisfied but (II) is, and (iii) (II) is not satisfied but (I) is. For case (i), we obtain what we wanted to show. 

Consider case (ii). As condition (I) is not satisfied, there exist a smallest $j$ such that
\be
    \beta_j + \sum_k \beta_k \geq d-j.
\ee
Now, for all $l \geq j$, let $\beta_l \to \beta_l - 1$ and $n_l \to n_l + 1$. Then, we can verify that conditions (I), (II) and (III$^\prime$) are now satisfied.

Finally, consider case (iii). As condition (II) is not satisfied, there exists a largest $j$ such that
\be
    \beta_j + \sum_k \beta_k < -j.
\ee
Then, for all $l \leq j$, let $\beta_l \to \beta_l + 1$ and $n_l \to n_l - 1$. We can verify that conditions (I), (II) and (III$^\prime$) are now satisfied. Thus, we have shown existence.

Now, we shall show uniqueness by contradiction. Suppose we have the following decompositions,
    \begin{align}
        2N\mathbf{x} &= \sum_i (n_i + \beta_i) \mathbf{v}^{(i)}
        \nonumber \\
        &=  \sum_i (m_i + \gamma_i) \mathbf{v}^{(i)}.
    \end{align}
where $m_i$ and $n_i$ are integers (with at least one $k$ such that $m_k\neq n_k$) and the conditions (I)-(III) are satisfied for both $\beta_i$ and $\gamma_i$. In other words, the $\{\beta_i\}$ we found earlier are not unique.

Because $\mathbf{v}^{(i)}$ are all linearly independent, $m_i+ \gamma_i = n_i+\beta_i$ for all $i$. It follows that, without loss of generality, we can assume that $m_i \geq n_i$ for all $i$. (If this is not true, than we can write two new decompositions such that $m_i^\prime = \max(m_i,n_i)$, $ n_i^\prime = \min(m_i,n_i)$ and $\gamma^\prime_i,\beta^\prime$ are picked from $\gamma_i,\beta_i$ accordingly.) Let $N_i = m_i - n_i \geq 0$ and we have at least one $k$ such that $N_k \neq 0$. Then, we can write $\beta_i = \gamma_i + N_i$.

Now, from conditions (I) and (II), we have the following for all $i$,
\begin{align}
    -1 \leq \gamma_i + \sum_k \gamma_k < 1,
    \label{eq:gammacondition}
    \\
    -1 \leq \beta_i + \sum_k \beta_k < 1.
    \label{eq:betacondition}
\end{align}

Now, using $\beta_i = \gamma_i + N_i$ in Eq.~\eqref{eq:betacondition}, we obtain
\be
    -1 \leq \gamma_i + \sum_k \gamma_k + N_i + \sum_k N_k < 1
    \label{eq:betaderived}
\ee
for all $i$. Now, for there to exist some $\{\gamma_i\}$, such that Eq.~\eqref{eq:betaderived} and Eq.~\eqref{eq:gammacondition} are both true for all $i$, we must have $N_i = 0$ for all $i$. We have reached a contradiction. Thus, the $\{n_i\}$ we found earlier must be unique.

\section{Hyper-volume of $\mathcal{T}$}
\label{app:volume}
Consider the following region
\be
    \mathcal{T}^{(D)}_{s} = \{ \mathbf{x} \in \mathbb{R}^{D} | \, \abs{x_i} \leq s \text{ and } \abs{x_i-x_j} \leq s \text{ for all } i,j\}.
\ee

Suppose we wish to find its hyper-volume (denoted as $\text{Vol}_{D}(\mathcal{T}^{(D)}_{s})$). To that goal, let $X_i$ be $D$ uniform i.i.d in $[-s,s]$. Then, let $M= \max_{i} X_i$ and $L = \min_{i} X_i$. Then, we can say that
\be
    \text{Vol}_{D}(\mathcal{T}^{(D)}_{s}) = (2s)^{D}\text{Prob}(M-L \leq s).
    \label{eq:probandvolume}
\ee

Now, consider the joint probability density function of $M$ and $L$,
\be
    \text{Prob}(M=m,L=l)=  \frac{D (D-1) (m-l)^{D-2}}{(2s)^D}.
\ee
Then, we have
\be
    \text{Prob}(M - L \leq 1) = \int_{I}   \frac{D (D-1) (m-l)^{D-2}}{(2s)^D} \,dm\, dl
\ee
where $I=\{(m,l)\,|\, 0\leq m-l\leq s \text{ and} -s \leq l \leq m \leq s \}$. Substituting this in Eq.~\eqref{eq:probandvolume}, we get
\be
    \text{Vol}_{D}(\mathcal{T}^{(D)}_{s}) =  D (D-1) \int_{I} (m-l)^{D-2} \,dm\, dl
\ee
\be
    \text{Vol}_{D}(\mathcal{T}^{(D)}_{s}) = (D+1)s^D.
\ee

Now, note that $\mathcal{T}^{(d-1)}_{1/(2N)}=\mathcal{T}+\partial\mathcal{T}$, where $\mathcal{T}$ is the region defined in Eq.~\eqref{eq:polytopedef}. Using the fact that including the boundary doesn't alter the hyper-volume, we finally get
\be
    \text{Vol}_{d-1}(\mathcal{T})=\frac{d}{(2N)^{d-1}}.
\ee

\bibliography{myref}


\end{document}